\documentclass[12pt]{amsart}
\input xy
\xyoption{all}
\usepackage{url} 
\usepackage{graphicx,enumerate}
\usepackage[latin1]{inputenc}
\usepackage{tkz-graph}
\usepackage{url}

\usepackage{graphicx,epsfig,amssymb}
\usepackage{subfigure}
\usepackage{amssymb}
\usepackage{mathrsfs}
\usepackage{amsmath}
\usepackage{amsfonts}
\usepackage{latexsym}

\theoremstyle{plain} 
\newtheorem{theorem}             {Theorem}  [section]

\newtheorem{proposition}[theorem]{Proposition}

\theoremstyle{definition}

\theoremstyle{remark}

\numberwithin{equation}{section}

\def\aand{\quad \textrm{and} \quad}

\begin{document}

\title[De-radicalization]{Global Stability in a  Mathematical Model of De-radicalization}


\author[]{Manuele Santoprete\\
                Fei Xu\\
}
                \address{Department of Mathematics\\
                         Wilfrid Laurier University\\
                         Waterloo, ON, Canada}%


\begin{abstract}
    Radicalization is the process by which people come to adopt increasingly extreme political, social or religious ideologies. When radicalization leads to violence,  radical thinking becomes a threat to national security.  De-radicalization programs are part of an effort to combat violent extremism and terrorism. This type of initiatives  attempt to alter violent extremists  radical  beliefs and violent  behavior  with the aim to  reintegrate them into society.
    In this paper we introduce a simple compartmental model suitable to describe  de-radicalization programs. 
    The population is divided into four compartments: $ (S) $ susceptible, $ (E) $ extremists, $ (R) $ recruiters, and $ (T) $ treatment.
    We calculate the basic reproduction number $ \mathcal{R} _0 $. For $ \mathcal{R}  _0<  1 $ the system  has one globally asymptotically stable  equilibrium where no extremist or recruiters are present. For $ \mathcal{R}  _0 >1 $ the system  has an additional equilibrium where extremists and recruiters are endemic to the population. A Lyapunov function is used to show that, for $ \mathcal{R}  _0 >1 $,  the endemic equilibrium is  globally asymptotically stable. We use numerical simulations to support our analytical results. Based on our model we asses strategies to counter violent extremism.

  \end{abstract}
\maketitle

\begin{center} \today \end{center}

\section{Introduction}
According to Horgan \cite{horgan2009walking} radicalization is  the social and psychological process of incrementally
    experienced commitment to extremist political or religious ideology. Radicalization can lead to violent extremism and therefore it has become a major concern for national security. Typical counterterrorism strategies fall into two categories:
    \begin{enumerate}
        \item Law enforcement approach: violent extremist  are investigated prosecuted and imprisoned
        \item Military approach: violent extremists  are killed or captured on the battlefield.
    \end{enumerate}
    Practitioners of counterterrorism agree that these approaches alone cannot break the cycle of violence \cite{selim2016approaches}.    The realization of the inadequacy of the counterterrorism approach has lead to different strategies, collectively known as countering violent extremism (CVE).  CVE is a collection of noncoercive activities whose aim is to intervene in an individual's path toward violent extremism, to interdict  criminal activity and to reintegrate those convicted of criminal activity into society.
   CVE programs can be divided into three broad classes \cite{ohalloran2017challenges,selim2016approaches,mastroe2016surveying,clutterbuck2015deradicalization}
    \begin{enumerate}
        \item {\it Prevention programs}, which  seek to prevent the radicalization process from occurring and taking hold in the first place;
        \item {\it Disengagement programs}, which  attempt to stop or control radicalization as it is occurring;
        \item {\it De-radicalization programs}, which  attempt to alter an individual  extremist beliefs and violent  behavior  with the aim to  reintegrate him into society. This type of programs often target convicted terrorists.
    \end{enumerate}
According to Horgan \cite{horgan2014makes} there are at least 15 publicly known de-radicalization programs from Saudi Arabia to Singapore, but there are likely twice as many.
In this paper we use a compartmental model to model de-radicalization programs.

The attempt to use quantitative methods in describing social dynamics is not new, and compartmental models have been used to study various aspect of social dynamics. For instance
Hayward introduced a model of church growth \cite{hayward1999mathematical}, Jeffs et al. studied a model of political party growth  \cite{jeffs2016activist}, Romero et al. analyzed a model for the spread of political third parties \cite{romero2011epidemiological} and Crisosto et al. studied   the growth of cooperative learning in large communities \cite{crisosto2010community}.  The dynamics of the spread of crime
was studied by McMillon, Simon and Morenoff \cite{mcmillon2014modeling} and by Mohammad and Roslan \cite{mohammad2017analysis}. A mathematical model of the spread of gangs was studied  by
Sooknanan, Bhatt, and Comissiong \cite{sooknanan2016modified}. The same authors studied the model for the interaction of police and gangs in \cite{sooknanan2013catching}. Castillo-Chavez and Song analyzed the transmission dynamics of fanatic behaviors
\cite{castillo2003models}, Camacho  studied a model of the interaction between terrorist and fanatic groups  \cite{camacho2013development}, Nizamani, Memon and Galam modelled   public outrage and the spread of violence  \cite{nizamani2014public}.
Compartmental models of radicalization were studied by Galam and Javarone \cite{galam2016modeling} and by  McCluskey and Santoprete  \cite{mccluskey2017bare}.

In this paper we build on the compartmental  model introduced in \cite{mccluskey2017bare}  by adding a treatment compartment.  This allows us to  consider de-radicalization in our  analysis.
 We divide the population into four compartments, $ (S) $ susceptible, $ (E) $ extremists, $ (R) $ recruiters, and $ (T) $ treatment (see Figure \ref{fig:1}).
Using this simple model, we attempt to  test the effectiveness of de-radicalization programs in countering violent extremism. This is an important issue since, at least  on the surface, these de-radicalization programs are promising. In fact, these programs  appear to be  cost effective, since they are far cheaper than indefinite detention \cite{horgan2014makes}.  However, the degree of government support for these programs  hinges on their efficacy and, unfortunately, indicators of success and measures of effectiveness remain elusive \cite{ohalloran2017challenges}.

As in \cite{mccluskey2017bare} we use  the basic reproduction number $ \mathcal{R} _0 $ to  evaluate strategies for countering violent extremism.
 We will show that for $ \mathcal{R}  _0<  1 $ the system  has a globally asymptotically stable  equilibrium with no individuals in the extremist, recruiter and treatment classes, and that for  $ \mathcal{R}  _0 >1 $ the system  has an additional equilibrium in which extremists and recruiters are endemic to the population. The latter equilibrium is globally asymptotically stable for $ \mathcal{R} _0 >1 $. Therefore, if $ \mathcal{R} _0 <  1 $ the ideology will be eradicated, that is, eventually  the number of recruiters and extremists will go to zero. When $ \mathcal{R} _0 >  1 $ the ideology will become endemic, that is, the recruiters and extremists will establish themselves in the population.
In our model the basic reproduction number is
\begin{equation}
     \mathcal{R} _0 = \frac{ \Lambda } { \mu }  \frac {\beta (c _E q _E + b _E q _R - \frac{(1-k)\delta p _E } { b _T } q _R )} { b _E b _R - c _E c _R - \frac{ (1-k)\delta} { b _T } (c _E p _R + b _R p _E) },
\end{equation}
where   $ \mu $ is the mortality rate of the susceptible population,   $ k $ is the fraction of successfully de-radicalized individuals, and $ \delta $ is the rate at which individuals leave the treatment compartment, so that $1/\delta $ is the average time spent in the treatment compartment.
The fraction of extremists and recruiters entering the treatment compartment are $ p _E $ and $ p _R $, respectively. Moreover, $ b _E= \mu + d _E + c _E + p _E  $ and $ b _R = \mu + d _R + c _R + p _R   $, where   $ d _E $ and $ d _R $ are the additional mortality rates of the extremists and recruiters, respectively.\footnote{In the context of the present model these can be viewed as the rates at which extremists and recruiters are imprisoned with life sentences.}  Other parameters are described in Section 2.
Note that, if $ p _E, p _R \to 0  $, then the basic reproduction number limits to the one of the bare-bones model studied in  \cite{mccluskey2017bare}.

One approach to dealing with extremism, which follows under the umbrella of counterterrorism, is to prosecute and imprison violent extremists. This approach was studied in \cite{mccluskey2017bare}  where it was shown that  increasing the parameters $ d _E $ and $ d _R $ resulted in a decrease in $ \mathcal{R} _0 $. A similar results holds for the model studied in this paper.
A different strategy consists in  improving the de-radicalization programs by either increasing the success rate $k $ or by increasing the rates $p _E $  and $ p _R $ at which extremists and recruiters enter the $ T $ compartment.  Since $ \mathcal{R} _0 $ is a decreasing function of $k, p_E$, and $p_R $ , increasing these parameters  decreases $ \mathcal{R} _0 $. Hence, according to our model, this is a successful strategy to counter violent extremism. Another option is to decrease $\delta $, which in turn decreases $\mathcal{R} _0 $. This approach is also viable because $ \mathcal{R} _0 $ is an increasing function of $ \delta $. A good way of thinking about this is to consider prison-based de-radicalization programs, in which case, decreasing $\delta $ corresponds to increasing $ \frac{ 1 } { \delta } $,  the average prison sentence.

Note that, in general,  it may not be easy to determine the values of  parameters because available data are scarce. It has been claimed,  however, that the de-radicalization  program in Saudi Arabia, has a rate of recidivism of  about 10-20\% \cite{horgan2014makes}, which gives an estimate for the value of $k$.

The paper is organized as follows. In Section 2we introduce the mathematical model. In Section 3 we find an equilibrium with no individuals in the extremists, recruiters and treatment compartments. We also  compute the basic reproduction number using the next generation method. In Section 4 we use Lyapunov functions to prove this critical point is globally  asymptotically stable for $ \mathcal{R} _0 < 1 $. In Section 5 we find another equilibrium point, the {\it endemic equilibrium},  and  we prove it is globally asymptotically stable for $ \mathcal{R} _0 >1 $. In Section 6 we present some  numerical simulations supporting our analytical results. The final section concludes the paper with a short summary and discussion of the results, limitations of our model and ideas for future research.

\section{Equations}\label{sec:2}


\tikzstyle{block} = [rectangle, draw, fill=blue!20,
    text width=5em, text centered, rounded corners, minimum height=4em]
\tikzstyle{line} = [draw, -latex]
We model the spread of extreme ideology as a contact process. We assume that within the full population there is a subpopulation potentially at risk of adopting the ideology. We partition this   subpopulation  into four compartments:
\begin{enumerate}
    \item $(S)$  Susceptible
    \item  $(E)$ Extremists
    \item $ (R) $ Recruiters
    \item $ (T) $ Treatment.
\end{enumerate}
Our model is based on the bare-bones mathematical model of radicalization introduced in \cite{mccluskey2017bare}.
Here, however,  we also include a   treatment compartment $ (T) $, to  describe de-radicalized individuals.
The transfer diagram for this system is given below.

  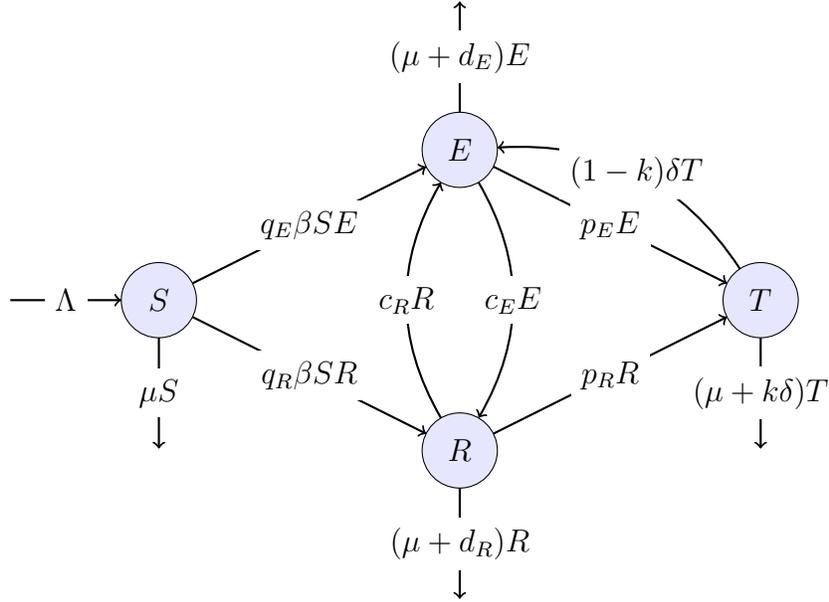
\begin{figure}[h]
\begin{tikzpicture}
\tikzset{VertexStyle/.style={circle,fill=blue!10,draw,minimum size=1cm,inner sep=5pt},
            }

   \Vertex[Math,L=S,x=0,y=0]{S}
   \Vertex[Math,L=E,x=4,y=2]{E}
   \Vertex[Math,L=R,x=4,y=-2]{R}
   \Vertex[Math,L=T,x=8,y=0]{T}
   \Vertex[empty,x=-2,y=0]{F0}
   \Vertex[empty,x=0,y=-2]{F1}
   \Vertex[empty,x=4,y=4]{F2}
   \Vertex[empty,x=4,y=-4]{F3}
   \Vertex[empty,x=8,y=-2]{F4}
   \tikzset{EdgeStyle/.append style={->}}
   \Edge[label = $q_E\beta SE$](S)(E)
   \Edge[label = $q_R\beta SR$](S)(R)
   \Edge[label = $c_R R$,style={bend left}](R)(E)
   \Edge[label = $c_E E$, style={bend left}](E)(R)
   \Edge[label = $p_R R$](R)(T)
   \Edge[label = $p_E E$](E)(T)
   \Edge[label = $(1-k)\delta T$, style={bend right}](T)(E)
   \Edge[label = $\Lambda$](F0)(S)
   \Edge[label = $\mu S$](S)(F1)
   \Edge[label = $(\mu+d_E) E$](E)(F2)
   \Edge[label = $(\mu+d_R) R$](R)(F3)
   \Edge[label = $(\mu+k\delta) T$](T)(F4)

    \end{tikzpicture}

\caption{Transfer diagram for the de-radicalization model.\label{fig:1}}
\end{figure}

We assume that susceptibles and recruiters interact according to a mass action law, and that the rate at
which susceptibles are recruited to adopt the extremist ideology is proportional to the number of interactions
that are occurring.  Thus, susceptibles are recruited at rate $\beta S R$, with a fraction $q_E$
entering the extremist class and a fraction $q_R = 1 - q_E$ entering the recruiter class. Extremists switch to the recruiter class with rate constant $ c_E $, while recruiters enter the extremist
class with rate constant $ c_R $. The natural death rate is proportional to the population
size, with rate constant $ \mu $. Extremists and recruiters have additional death rates $ d  _E $ and $d _R$, respectively. These rates account for individuals that are imprisoned for life or killed.
To consider individuals that undergo de-radicalization program, extremists and recruiters are made to  enter the treatment compartment at rate constants $ p _E $ and $ p _R $ respectively.
The rate at which a treated individual leaves the compartment $ T $ is $ \delta $.
A fraction $ k \in [0,1]$ of treated individuals is removed, since we assume that successfully treated individuals are permanently de-radicalized. This seems to be  a reasonable assumption since, according to Horgan \cite{horgan2014makes}, individuals who leave terrorism behind have a low chance of re-engagement.
The  fraction of individuals for which the de-radicalization program fails is $1-k$. These individuals enter   the extremist class $ E $  after being treated.
Thus, the radicalization model consists of the following differential equations together with
non-negative initial conditions:
\begin{equation}\label{eqn:modela}
 \begin{aligned}
 S' & = \Lambda - \mu S - \beta SR		\\
 E' & = q_E\beta SR - (\mu + d_E + c_E+ p _E ) E + c_R R +(1-k)\, \delta T		\\
 R' & = q_R \beta SR + c_E E - (\mu + d_R + c_R+p _R ) R\\
T' & = p _E E + p _R R - (\mu +\delta)T
 \end{aligned}
\end{equation}
where $ q _E + q _R = 1 $, $ q _E , q _R \in [0,1] $. For simplicity denote $ b _E = \mu + d_E + c_E+ p _E $, $ b _R =  \mu + d_R + c_R+ p _R  $ and $ b _T = \mu + \delta $, then system \eqref{eqn:modela} takes the following form:
\begin{equation}\label{eqn:modelb}
 \begin{aligned}
 S' & = \Lambda - \mu S - \beta SR		\\
 E' & = q_E\beta SR - b _E  E + c_R R +(1-k) \, \delta T		\\
 R' & = q_R \beta SR + c_E E - b _R  R\\
T' & = p _E E + p _R R - b _T T
 \end{aligned}
\end{equation}

\begin{proposition}\label{prop:invariant_region}
 The  region  $ \Delta = \left \{ (S, E , R,T) \in \mathbb{R}  ^4 _{ \geq 0 } : S +E+R+T \leq \frac{ \Lambda } { \mu } \right \} $ is a compact  positively invariant set   for the flow of \eqref{eqn:modela} (i.e. all  solutions starting in $ \Delta  $ remain in $ \Delta   $ for all $ t> 0 $). Moreover, $ \Delta $ is attracting within $ \mathbb{R}  ^4 _{ \geq 0 } $ (i.e.  solutions starting outside $ \Delta $ either enter or approach $ \Delta $ in the limit.

\end{proposition}
\begin{proof}
    It is trivial to check that $ \Delta $ is compact. We first show that $ \mathbb{R}  ^4 _{ \geq 0 } $ is positively invariant  by  checking  the direction of the vector field along the boundary of  $ \mathbb{R}  ^4 _{ \geq 0 } $. Along $ S = 0 $ we have $ S ' = \Lambda>0 $ so the vector field points inwards.
    Along $ E = 0 $ we have $ E' = q _E \beta S R + c _R R + (1 - k) \delta T \geq 0 $, provided $ R, S , T \geq 0 $. Moreover, along $ R = 0 $, we have that  $ R' = c _E E \geq 0 $ provided $ E \geq 0 $. Moreover,
  along $ T = 0 $ we have $ p _E E + p _R R \geq 0 $, provided $ E, R \geq 0 $. This shows that  $ \mathbb{R}  ^4 _{ \geq 0 } $ is positively invariant by Proposition 2.1 in \cite{haddad2010nonnegative}.
Now let $ N = S + E + R + T $, then
 \[ S'+E'+R'+T'=\Lambda - \mu N - d _E E - d _R R -  k\delta T \leq \Lambda - \mu N.\]
Using a standard comparison theorem, it follows that
\begin{equation}\label{eqn:ineq}
    N (t) \leq \left( N (0) - \frac{ \Lambda } {\mu } \right) e ^{ - \mu t } + \frac{ \Lambda } { \mu },
\end{equation}
for $ t \geq 0 $. Thus, if $ N (0) \leq \frac{ \Lambda } { \mu } $, then $ N (t) \leq \frac{ \Lambda } { \mu } $ for all $ t \geq 0 $. Hence, the set  $ \Delta $ is positively invariant. Furthermore, it follows from \eqref{eqn:ineq} that $ \limsup_{t\to \infty } N \leq \frac{ \Lambda } { \mu } $, demonstrating that $ \Delta $ is attracting within $ \mathbb{R}  ^4 _{ \geq 0 } $.
\end{proof}

\section{Radicalization-free equilibrium and basic reproduction number $ \mathcal{R} _0 $}\label{sec:3}

If $ E = R =T= 0 $, then an equilibrium is given by $ x _0 = \left( S _0 , E _0 , R _0 , T _0  \right)  = \left( \frac{ \Lambda } { \mu } ,0,0,0 \right) $.

The basic reproduction number $\mathcal{R}_0 $  is the spectral radius of the next generation matrix $G $
calculated at $ x _0 $. $\mathcal{R}_0$ can be calculated as follows (see \cite{van2002reproduction} for
more details).  In our case the infected compartments are $E,R,T  $.
The next generation matrix is given by $ G  = F V ^{ - 1 }$ with
\[
    F = \begin{bmatrix}
        \frac{ \partial \mathcal{F}  _E } { \partial E }  & \frac{ \partial \mathcal{F}  _E } { \partial R }  &  \frac{ \partial \mathcal{F}  _E } { \partial T }\\[1em]
        \frac{ \partial \mathcal{F}  _R } { \partial E }  & \frac{ \partial \mathcal{F}  _R } { \partial R } &  \frac{ \partial \mathcal{F}  _R } { \partial T }\\[1em]
        \frac{ \partial \mathcal{F}  _T } { \partial E }  & \frac{ \partial \mathcal{F}  _T } { \partial R } &  \frac{ \partial \mathcal{F}  _T } { \partial T }\\

    \end{bmatrix} (x_0)
  \aand
  V = \begin{bmatrix}
        \frac{ \partial \mathcal{V}  _E } { \partial E }  & \frac{ \partial \mathcal{V}  _E } { \partial R }  &  \frac{ \partial \mathcal{V}  _E } { \partial T }  \\[1em]
        \frac{ \partial \mathcal{V}  _R } { \partial E }  & \frac{ \partial \mathcal{V}  _R } { \partial R } &  \frac{ \partial \mathcal{V}  _R } { \partial T }\\[1em]
         \frac{ \partial \mathcal{V}  _T } { \partial E }  & \frac{ \partial \mathcal{V}  _T } { \partial R } &  \frac{ \partial \mathcal{V}  _T } { \partial T }
    \end{bmatrix} (x_0).
\]
Here, $ \mathcal{F}  _E $, $ \mathcal{F}  _R $ and $ \mathcal{F}  _T $  are the rates of appearance of newly radicalized individuals
in the classes $ E $, $ R $, and $ T $,  respectively. Let $ \mathcal{V}  _j = \mathcal{V}  _j ^{ - } - \mathcal{V}  _j ^{ + } $,  with $ \mathcal{V}  _j ^{ + } $  is the rate of
transfers of individuals into class $ j $    by all other means, and $ \mathcal{V}  _j ^{ - } $ is the rate of transfers of individuals out of class $j$, where $ j \in \{ E,R,T \}$.
In our case
\[
    \mathcal{F} = \begin{bmatrix}
        \mathcal{F}  _E   \\
        \mathcal{F}  _R\\
        \mathcal{F}  _T
    \end{bmatrix} =\beta S
    \begin{bmatrix}
       q _E R\\
       q _R R\\
       0
  \end{bmatrix}
\]
and
\[
\mathcal{V} =  \begin{bmatrix}
          \mathcal{V}  _E   \\
          \mathcal{V}  _R
    \end{bmatrix} =
    \begin{bmatrix}
     b _E E - c _R R - (1 - k) \delta T    \\
     b _R R - c _E E\\  
     b _T T - (p _E E + p _R R)
  \end{bmatrix}.
\]
Hence
\[
    F =\beta S _0 \begin{bmatrix}
         0&   q _E & 0 \\
         0 & q _R  & 0\\
         0 & 0 & 0
    \end{bmatrix}
\aand
   V =  \begin{bmatrix}
       b _E & - c _R & - \alpha _E  \\
       - c _E & b _R  & 0\\
       - p _E & - p _R & b _T
  \end{bmatrix}.
\]

Therefore,
\begin{align*}
    G & =\frac{S _0 \beta
           }{\tilde D}
 \begin{bmatrix}
        0  &  q_E & 0  \\
        0& q _R & 0 \\
        0 & 0 & 0
    \end{bmatrix}
 \begin{bmatrix}
       - b _R b _T  & -(\alpha _E p _R + c _R b _T )& \alpha _E b _R  \\
      -c _E b _T  & \alpha _E p _E - b _E b _T & - \alpha _E c _E \\
     - b _R p _E - c _E p _R & - b _E p _R - c _R p _E & - b _E b _R + c _R c _E
  \end{bmatrix}\\
 & = \frac{ \beta S _0 } { \tilde D } \begin{bmatrix}
    -q _E c _E b _T   & q _E(\alpha _E p _E - b _E b _T ) & - q _E \alpha _E c _E  \\
   -q _R c _E b _T   & q _R( \alpha_E p _E - b _E b _T) & - q _R \alpha _E c _E \\
  0 & 0 & 0
   \end{bmatrix},
   \end{align*}
where $\tilde D =   \alpha _E (b _R p _E + c _E p _R) + b _T (c _E c _R - b _E b _R)$.
Note that $F$ has rank $1$ and so the same is true for $G$.  Since two eigenvalues of $G$ are
zero the spectral radius is equal to the absolute value of the remaining eigenvalue.  Since the
trace is equal to the sum of the eigenvalues and there is only one non-zero eigenvalue, we
see that the spectral radius of $G$ is equal to the absolute value of the trace (which happens
to be positive).  Thus,
\begin{equation}				\label{Rzero}
     \mathcal{R} _0 =  \frac {\beta S _0  (c _E q _E + b _E q _R - \frac{ \alpha _E p _E } { b _T } q _R )} { b _E b _R - c _E c _R - \frac{ \alpha _E } { b _T } (c _E p _R + b _R p _E) }. 
\end{equation}

\section{Global Asymptotic Stability of  $ x _0 $ for $ \mathcal{R} _0 < 1 $}\label{sec:4}
In this section, we investigate the stability of the critical point $ x _0 $. The next generation method provides us with information on the local stability: $ x _0 $ is locally asymptotically stable for  $ \mathcal{R} _0 < 1 $ and unstable if $ \mathcal{R} _0 > 1 $.
The global asymptotical stability of $ x _0 $, instead,  is given by the following theorem.
\begin{theorem}\label{thm:x0}
   If $ \mathcal{R} _0 \leq 1 $ then $ x _0 $ is globally asymptotically stable on  $ \mathbb{R}  ^4 _{ \geq 0 }$.
\end{theorem}
\begin{proof}
   Consider the following $ C^1 $ Lyapunov function
   \[
       U = b _T c _E E + (b _T b _E - \alpha _E p _E) R + \alpha _E c _E T.
   \]
 Evaluating  the time derivative of $ U $ along the trajectories of \eqref{eqn:modelb} yields
\begin{align*}
    U'  =&  b _T c _E E' + (b _T b _E - \alpha _E p _E) R' + \alpha _E c _E T'\\
    =&  b _T c _E (q _E \beta SR - b _E E + c _R R + \alpha _E T) + (b _T b _E - \alpha _E p _E) (q _R \beta SR + c _E E - b _R R)\\
    & + \alpha _E c _E (p _E E + p _R R - b _T T)  \\
    = &  b _T \left[ \beta (q _E c _E + q _R b _E - q _R \frac{ \alpha _E } { b _T } p _E) S -
        \left( b _E b _R - c _E c _R - \frac{ \alpha _E } { b _T } (p _E b _R + c _E p _R)\right)  \right] R    \\
    = &   b _T D\left[ \beta (q _E c _E + q _R b _E - q _R \frac{ \alpha _E } { b _T } p _E) \frac{  S}{D} - 1
         \right] R    \\
    = &  b _T D\left[ \mathcal{R} _0 \frac{  S}{S _0 } - 1\right] R
\end{align*}
where $ D = b _E b _R - c _E c _R - \frac{ \alpha _E } { b _T } (p _E b _R + c _E p _R) $. It follows from $ S \leq S _0 = \frac{ \Lambda } { \mu } $ that
\begin{align*}
   U ' \leq & \,   b _T D\left[ \mathcal{R} _0  - 1\right] R
\end{align*}
which implies that $ U'
\leq 0 $ if $ \mathcal{R} _0 \leq 1 $. Furthermore, $ U ' =0 $ if and only if $ \mathcal{R} _0 = 1 $ or  $ R = 0 $.  Let
\[ Z= \{ (S,E, R,T) \in \Delta |\,U' = 0 \}.\]
We claim that the largest invariant set contained in $ Z $ is $ x _0 $.
In fact,  any entire solution $ (S(t), E(t), R(t) ,T(t)) $ contained in $ Z $  must have $ R (t) \equiv 0 $ as a consequence of the expression for $ U' $ given above. Moreover, from the second and third line in \eqref{eqn:modelb} it follows that  $ E (t)\equiv 0 $ and  $ T (t) \equiv 0 $. Substituting $ R = T = 0 $ in the first line of \eqref{eqn:modelb} gives a differential equation with solution $ S = \left(  S (0) - \frac{ \Lambda } { \mu }\right)   e ^{ - \mu t } + \frac{ \Lambda } { \mu }$. Clearly, if $ S(0) \leq \frac{ \Lambda } { \mu } $, then $ S\to - \infty $ as $ t \to - \infty $ and the corresponding entire solution is not contained in $ Z $. It follows that $ S (0) = \frac{ \Lambda } { \mu } $, which proves the claim.

Since $ \Delta $ is  positively invariant with respect to \eqref{eqn:modelb}  LaSalle's invariance principle (\cite{khalil1996noninear} Theorem 4.4 or \cite{la1976stability} Theorem 6.4) implies that  all trajectories that start in $ \Delta $ approach  $ x _0 $  when $ t \to \infty $. 
This together with the fact that $ x _0 $ is Lyapunov stable (in fact is locally asymptotically stable by the next generation method), prove  that $ x _0 $ is globally asymptotically stable in $ \Delta $.
Since  $ \Delta $  is an attracting set within $ \mathbb{R}  ^4 _{ \geq 0 } $  the stability is also global in $ \mathbb{R}  ^4 _{ \geq 0 } $.


\end{proof}
\section{Global Asymptotic Stability of the Endemic Equilibrium}\label{sec:5}
In this section,  we show that if $ \mathcal{R} _0 >1 $, then \eqref{eqn:modelb} has a unique endemic equilibrium. We then study the global asymptotic stability of such equilibrium using Lyapunov functions.

An endemic equilibrium $ x ^\ast = (S ^\ast , E ^\ast , R ^\ast, T ^\ast ) \in \mathbb{R} ^4_{>0} $ of \eqref{eqn:modelb} is an equilibrium in which at least one of $ E ^\ast , R ^\ast $ and $ T ^\ast $ is nonzero.
To find the endemic equilibria of \eqref{eqn:modelb} we first set $ T' = 0 $, from which we obtain $ T = \frac{ p _E } { b _T } E ^\ast  + \frac{ p _R } { b _T } R ^\ast$. Using the expression above for $ T ^\ast $, setting $ E' = R' = 0 $ and  treating $ S ^\ast $ as a parameter
yields the linear system
\begin{equation}
    \begin{bmatrix}
    - b _E + \frac{ p _E } { b _T } (1 - k)  \delta   & q _E \beta S ^\ast + c _R + \frac{ p _R } { b _T } (1 - k) \delta \\
   c _E & q _R \beta S ^\ast - b _R 
  \end{bmatrix} \begin{bmatrix}
    E ^\ast   \\
     R ^\ast
\end{bmatrix}  = \begin{bmatrix}
    0\\
    0
\end{bmatrix}
\end{equation}
In order to have non-zero solutions for $ E ^\ast $ and $ R ^\ast $, the coefficient matrix must have determinant zero. This gives
\begin{equation}\label{eqn:Sstar}
    S ^\ast = \frac{ b _E b _R - c _E c _R - \frac{ \alpha _E } { b _T } (c _E p _R + b _R p _E) } {
       \beta (c _E q _E + b _E q _R - \frac{ \alpha _E p _E } { b _T } q _R )} = \frac{ \Lambda } { \mu } \frac{ 1 } { \mathcal{R} _0 },
\end{equation}
where $ \alpha _E = (1 - k) \delta $.
Solving the third equation for $ E ^\ast $ yields
\[
    E ^\ast = \omega R ^\ast, \quad \mbox{ with } \quad \omega = \frac{ b _R - q _R \beta S ^\ast } { c _E }.
\]
Next, taking the last line in \eqref{eqn:modelb} and solving for $ T ^\ast $ gives
\[
    T ^\ast = \frac{ p _E \omega + p _R } { b _T } R ^\ast.
\]
Substituting this last expression in the first line of \eqref{eqn:modelb} we obtain
\[
    R ^\ast = \frac{ \Lambda - \mu S ^\ast } {\beta S ^\ast } = \frac{ \mu}{\beta } (\mathcal{R} _0 - 1).
\]
It follows that  a meaningful  endemic equilibrium with positive  $ S ^\ast , E ^\ast , R ^\ast $, and $ T ^\ast $  exists if and only if $ \mathcal{R} _0 > 1 $. When the endemic equilibrium exists, there is only one, denoted by   $ x ^\ast = (S ^\ast , E ^\ast , R ^\ast, T ^\ast) $, where
\begin{equation}
    \begin{aligned}
          S ^\ast & = \frac{ \Lambda } { \mu } \frac{ 1 } { \mathcal{R} _0 },\\
          E ^\ast & = \omega R ^\ast\\
          R ^\ast & = \frac{ \mu}{\beta } (\mathcal{R} _0 - 1)\\
          T ^\ast & = \frac{ p _E \omega + p _R } { b _T } R ^\ast.
    \end{aligned}
\end{equation}

\begin{theorem}\label{thm:x_ast} If $ \mathcal{R} _0 >1 $, then the endemic equilibrium $ x ^\ast $  of \eqref{eqn:modelb} is globally asymptotically stable in $ \mathbb{R}  ^4 _{ >0 }  $.
\end{theorem}

\begin{proof}
We study the global stability of $ x ^\ast $ by considering the Lyapunov function
\[
    V = S ^\ast g \left( \frac{ S } { S ^\ast } \right) + a _1  \, E ^\ast g \left( \frac{ E } { E ^\ast } \right) + a _2 \, R ^\ast g \left( \frac{ R } { R ^\ast } \right)+ a _3  \, T^*g \left( \frac{ T } { T ^\ast } \right)
\]
where $ g (x) = x - 1 - \ln{x}$. Clearly $ V $ is $ C ^1 $, $ V (x ^\ast) = 0 $, and $ V>0 $ for
any $ p \in \mathbb{R} ^4 _{ >0 }$ such that  $ p \neq x ^\ast $.

Differentiating $ V $  along solutions of \eqref{eqn:modelb} yields
\begin{align*}
   V'  = &  \left( 1 - \frac{ S ^\ast } { S } \right) S' + a_1\left( 1 - \frac{ E ^\ast } { E } \right) E' + a_2 \left( 1 - \frac{ R ^\ast } { R }\right) R' + a_3\left( 1 - \frac{ T ^\ast } { T } \right) T'\\
  =&   \left( 1 - \frac{ S ^\ast } { S } \right) [\Lambda  - \mu S - \beta S R ]+   a_1 \left( 1 - \frac{ E ^\ast } { E } \right)[ q _E \beta S R - b _E E + c _R R + \alpha _E T]\\
  & +  a_2 \left( 1 - \frac{ R ^\ast } { R }\right)[q _R \beta S R + c _E E - b _R R]+ a_3 \left( 1 - \frac{ T ^\ast } { T } \right) [p _E E + p _R R - b _T T]\\
   =&  C -(\mu +a _2 \beta q _R R ^\ast )S +(a _1  q _E + a _2  q _R - 1) \beta SR +(-a _1 b _E + a _2 c _E + a _3 p _E )E\\
  & +(S ^\ast \beta + a _1 c _R - a _2 b_R + a _3 p _R )R + ( a _1 \alpha _E - a _3 b _T ) T -\Lambda \frac{ S ^\ast } { S } - a _3 p _E \frac{ T ^\ast } { T }E - a _2 c _E \frac{ R ^\ast } { R }E\\
 &  - a _3 p _R \frac{ T ^\ast } { T }R  - a _1 \alpha _E E ^\ast \frac{ T } { E } - a _1 c _R E ^\ast \frac{ R } { E } - a _1 \beta q _E E ^\ast \frac{ SR } { E }
\end{align*}
where $ C = \Lambda + \mu S ^\ast + a _1 b _E  E ^\ast + a _2 b _R R ^\ast + a _3 b _T T ^\ast $.
For simplicity, denote $ w = \frac{ S } { S ^\ast } $, $ x = \frac{ E } { E ^\ast } $, $ y = \frac{ R } { R ^\ast } $, and $ z = \frac{ T } { T ^\ast } $.
Then,
\begin{align*}
   V' = &   C -(\mu +a _2 \beta q _R R ^\ast )S ^\ast w  +(a _1  q _E + a _2  q _R - 1) \beta S ^\ast R^\ast wy +(-a _1 b _E + a _2 c _E + a _3 p _E )E ^\ast x\\
  & +(S ^\ast \beta + a _1 c _R - a _2 b_R + a _3 p _R )R ^\ast y + ( a _1 \alpha _E - a _3 b _T ) T ^\ast z  -\Lambda \frac{ 1 } { w } - a _3 p _E E^\ast  \frac{ x } { z }\\
  & - a _2 c _E E^\ast \frac{ x } { y }  - a _3 p _R R ^\ast \frac{ y } { z }  - a _1 \alpha _E T ^\ast \frac{ z } { x } - a _1 c _R R ^\ast \frac{ y } { x } - a _1 \beta q _E S ^\ast R ^\ast  \frac{ wy } { x }  :=G(w,x,y,z).
\end{align*}
As in \cite{li2012algebraic}, we define a set  $ \mathcal{D} $  of the above terms as follows
\[
    \mathcal{D}  = \left\{w,x,y,z,wy,\frac{ 1 } { w } , \frac{ x } { z } , \frac{ x } { y }, \frac{ y } { z }  , \frac{ z } { x } , \frac{ y } { x } , \frac{ w y } { x }  \right\}.
\]

There are at most five subsets associated with $ \mathcal{D} $ such that the product of all functions within each subset is equal to one, given by 
\[
    \left\{ w, \frac{ 1 } { w }  \right\},\left\{ \frac{x}{y}, \frac{ y } { x }  \right\}
    \left\{ \frac{x}{z}, \frac{ z } { x }  \right\},
    \left\{  \frac{ z } { x },\frac{ y } { z } , \frac{ x } { y }   \right\}
\left\{  \frac{ 1 } { w },\frac{ wy } { x } , \frac{ x } { y }   \right\}.
\]
We associate to these subsets of variables the following terms
\begin{align*} &  \left( 2 - w - \frac{ 1 } { w } \right) , \left( 2 - \frac{ x } { y } - \frac{ y } { x } \right),\left( 2 - \frac{ x } { z } - \frac{ z } { x } \right), \left( 3 - \frac{ z } { x } - \frac{ y } { z } - \frac{ x } { y } \right), \left( 3 - \frac{ 1 } { w } - \frac{ x } { y } - \frac{ w y } { x } \right).
\end{align*}
Following the method used in  \cite{li2012algebraic,lamichhane2015global} we  constructs a Lyapunov function as a linear combination of the terms above:
\begin{equation}
\begin{aligned} \label{eqn:derivative}
    H(w,x,y,z) =&  b _1 \left( 2 - w  - \frac{ 1 } { w } \right) + b _2 \left( 2 - \frac{ x } { y } - \frac{ y } { x } \right)
    + b _3 \left( 2 - \frac{ x } { z } - \frac{ z } { x } \right)\\
   &  + b _4 \left( 3 - \frac{ z } { x } - \frac{ y } { z } - \frac{ x } { y } \right) +
 b _5 \left( 3 - \frac{ 1 } { w } - \frac{ x } { y } - \frac{ w y } { x } \right),
\end{aligned}
\end{equation}
    where the coefficients $ b _1 , \ldots , b _5  $ are left unspecified.
   We want to determine suitable parameters $ a _i >0 $ ($ i = 1,2,3 $ ) and $ b _k \geq 0 $ ($ i =1,\ldots, 5 $ ) such that $ G (w,x,y,z) = H (w,x,y,z)$. Equating the coefficient of like terms in $G$ and $H$ gives the following equations:
   \begin{align*}
       w^0:&\quad 2(b _1 + b _2 + b _3 )+3(b _4 + b _5 )= C\\
       w:&\quad b _1 =( \mu + a _2 \beta q _R R ^\ast )S ^\ast \\
       wy:&\quad a _1 q _E + a _2 q _R - 1 = 0 \\
       x:&\quad - a _1 b _E + a _2 c _E + a _3 p _E = 0 \\
       y:&\quad S ^\ast \beta + a _1 c _R- a _2 b _R + a _3 p _R = 0 \\
       z:&\quad a _1 \alpha _E - a _3 b _T=0\\ 
       w^{-1}: &\quad b _1 + b _5 = \Lambda\\
       x z ^{ - 1 }:  &\quad b _3 = a _3 p _E E ^\ast \\
       xy ^{ - 1 }:  &\quad b _2 + b _4 + b _5 = a _2 c _E E ^\ast \\
       y z ^{ - 1 }: &\quad b _4 = a _3 p _R R ^\ast \\
       z x ^{ - 1 }:  &\quad b _3 + b _4 = a _1 \alpha _E T ^\ast \\
       y x ^{ - 1 } :  &\quad b _2 = a _1 c _R R ^\ast \\
        wy x ^{ - 1 } : &\quad b _5 = \beta a _1 q _E S ^\ast R ^\ast.
   \end{align*}
If we take $ (S ^\ast , E ^\ast , R ^\ast , T ^\ast) $  at the endemic equilibrium then the linear system above is consistent and has a unique solution with
\begin{align*}
    a _1 & = \frac{c_E}{ c _E q _E+ b _E q _R - \frac{ p _E } { b _T } \alpha _E q _R}\\
     a _2 & = \frac{ 1 } { q _R } -  \frac{ \frac{q_E}{q_R}c _E  }{ c _E q _E+ b _E q _R - \frac{ p _E } { b _T } \alpha _E q _R}\\
     a _3 & = \frac{\frac{c_E\alpha_E}{b_T}}{ c _E q _E+ b _E q _R - \frac{ p _E } { b _T } \alpha _E q _R},
\end{align*}
and with $ b _1 , \ldots , b _5 >0 $.
By  the arithmetic mean-geometric mean inequality each of the terms in \eqref{eqn:derivative} is less than or equal to zero. Furthermore,
\[ \mathcal{M} = \left\{ (S,E,R,T) \in \mathbb{R}  ^4 _{ >0 } |\, \frac{ d V } { dt } = 0 \right\}=\left\{ (S,E,R,T) \in \mathbb{R}  ^4 _{ >0 } |\, S = S ^\ast,
        \frac{ E } { E ^\ast } = \frac{ R } { R ^\ast } = \frac{ T } { T ^\ast} \right\}.\]
We claim that the  largest invariant set in $ \mathcal{M} $ is the set consisting of the endemic equilibrium $ x ^\ast $.  In fact, let $ (S(t),E(t),R(t),T(t)) $ be a  complete orbit in $ \mathcal{M} $, then
\[
    0=S'= (S ^\ast) ' =  \Lambda - \mu S ^\ast - \beta S ^\ast R,
\]
which implies that
\[
    R = \frac{ \Lambda - \mu S ^\ast } { \beta S ^\ast } = R ^\ast.
\]
Therefore, $x ^\ast =  (S(t),E(t),R(t),T(t)) $.
By LaSalle's invariance principle \cite{khalil1996noninear,la1976stability}, we deduce that all solutions of \eqref{eqn:modelb} that start in $ \mathbb{R}  ^4 _{ >0 }$ limit to  $  x ^\ast $. The fact that $ x ^\ast $ is globally asymptotically stable follows from a corollary to the invariance principle \cite{khalil1996noninear,la1976stability}.

\end{proof}

\renewcommand{\floatpagefraction}{0.1}
\begin{figure}[!p]
\begin{center}
\subfigure[]{\rotatebox{0}{\includegraphics[width=0.45 \textwidth,
height=45mm]{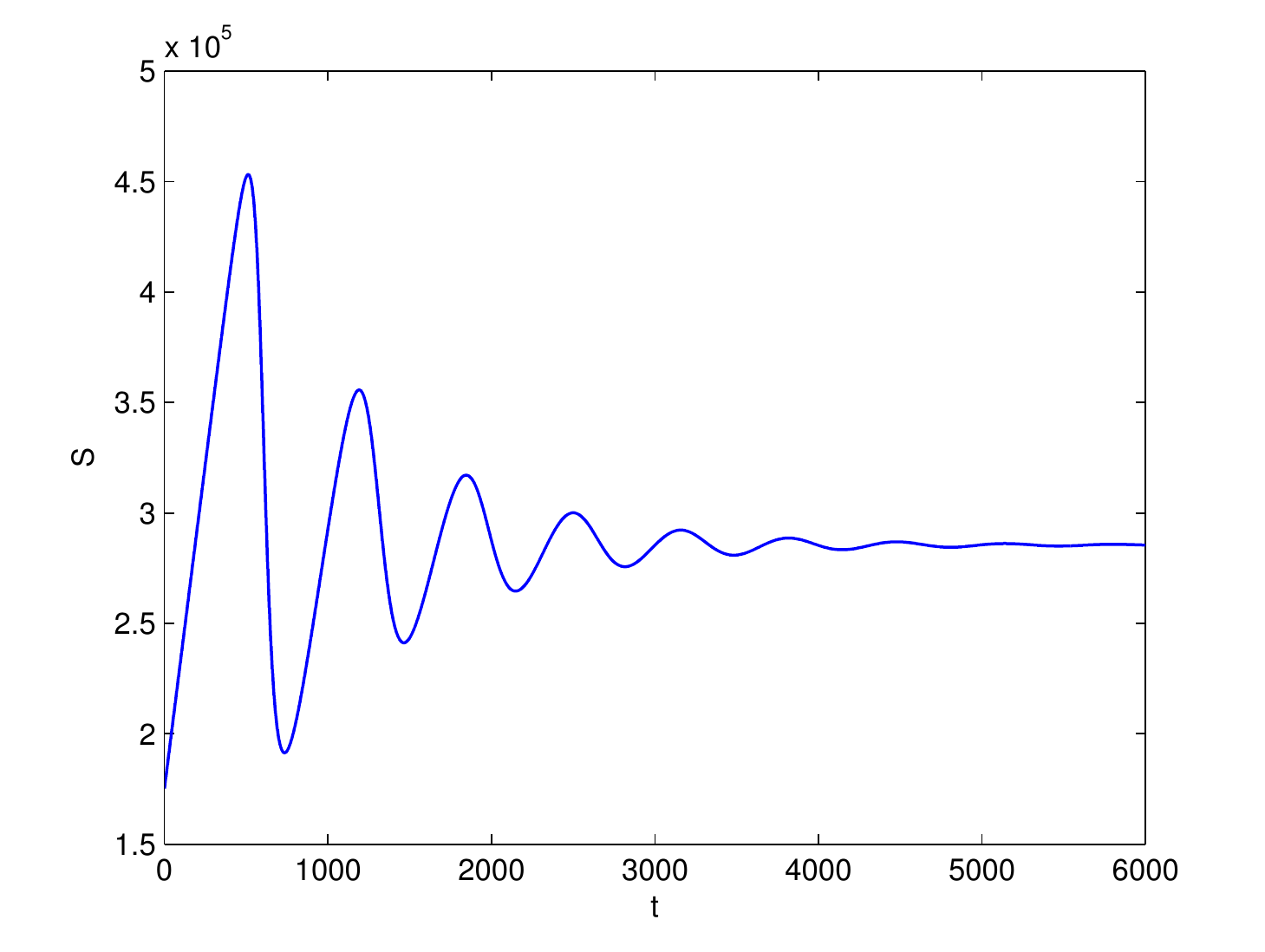}}}
\subfigure[]{\rotatebox{0}{\includegraphics[width=0.45 \textwidth,
height=45mm]{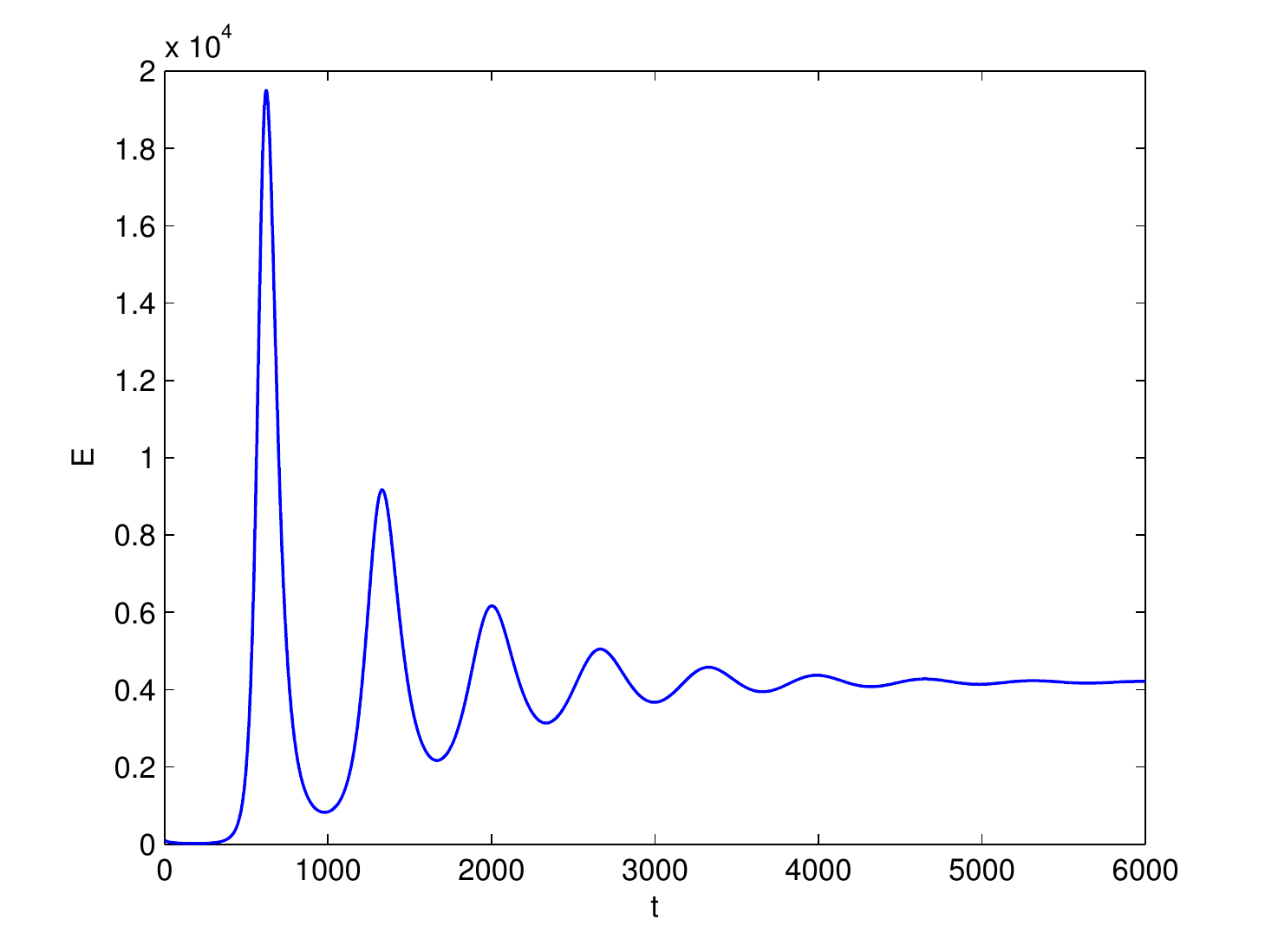}}}
\subfigure[]{\rotatebox{0}{\includegraphics[width=0.45 \textwidth,
height=45mm]{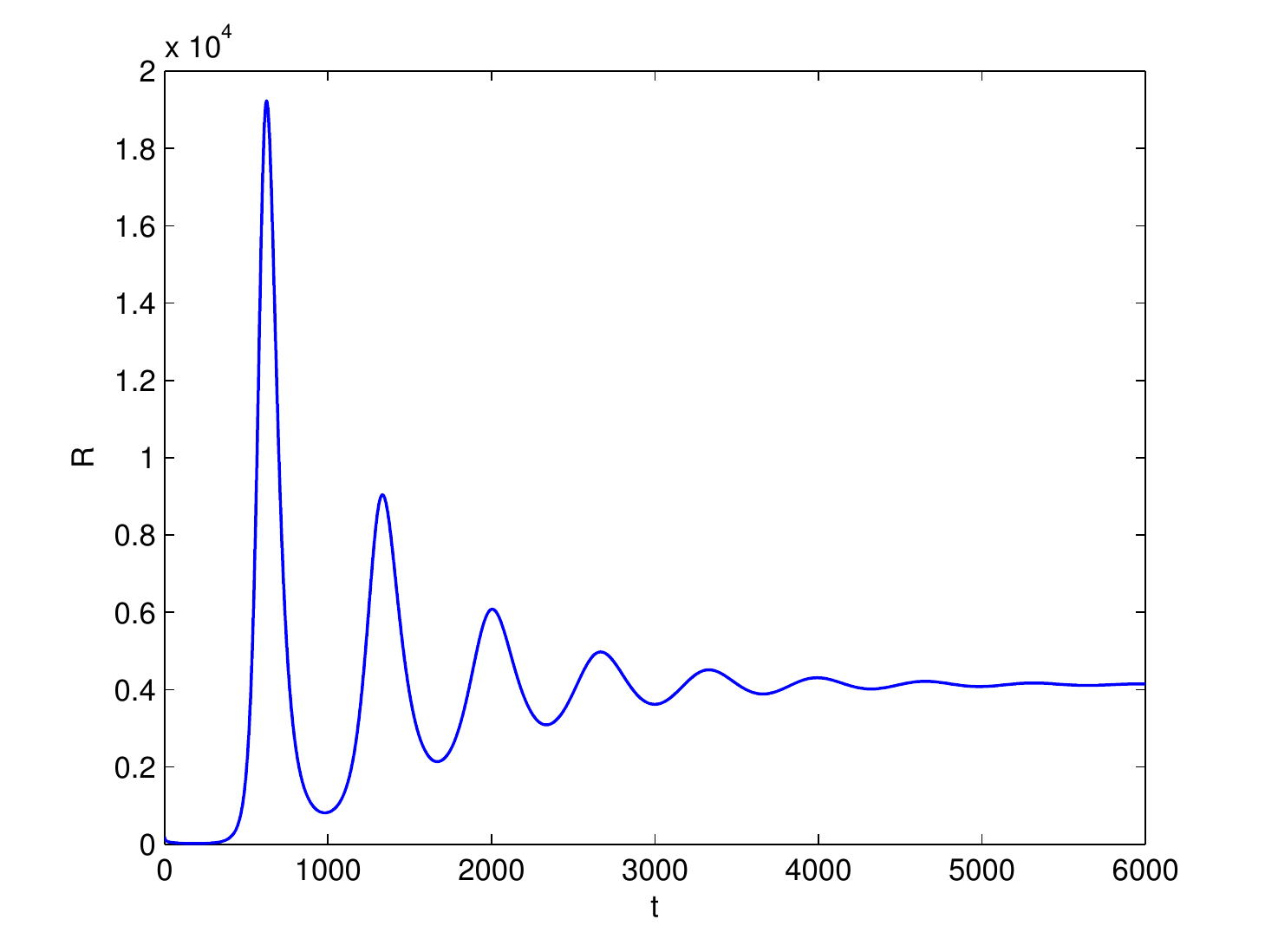}}}
\subfigure[]{\rotatebox{0}{\includegraphics[width=0.45 \textwidth,
height=45mm]{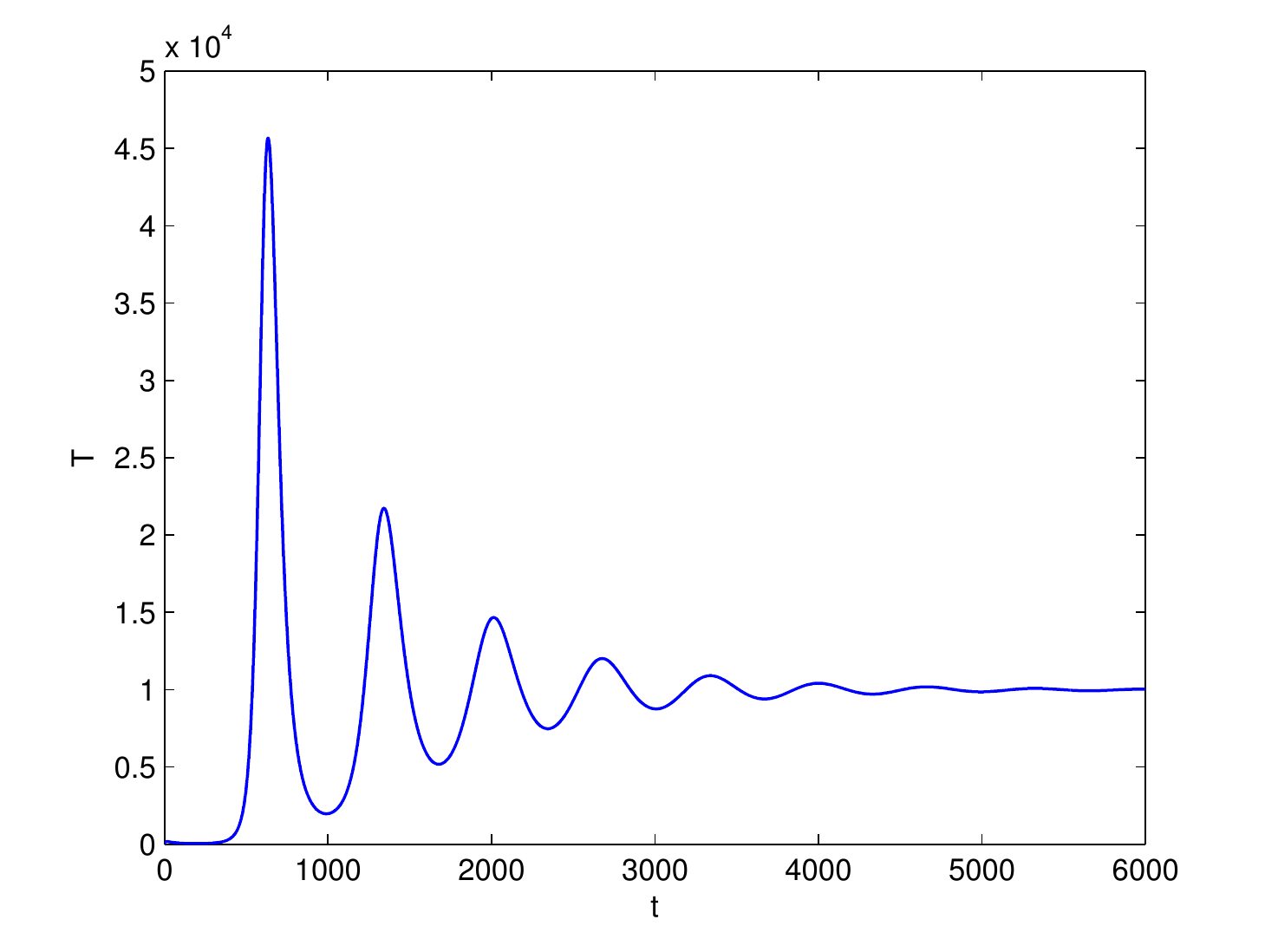}}}
\subfigure[]{\rotatebox{0}{\includegraphics[width=0.45 \textwidth,
height=45mm]{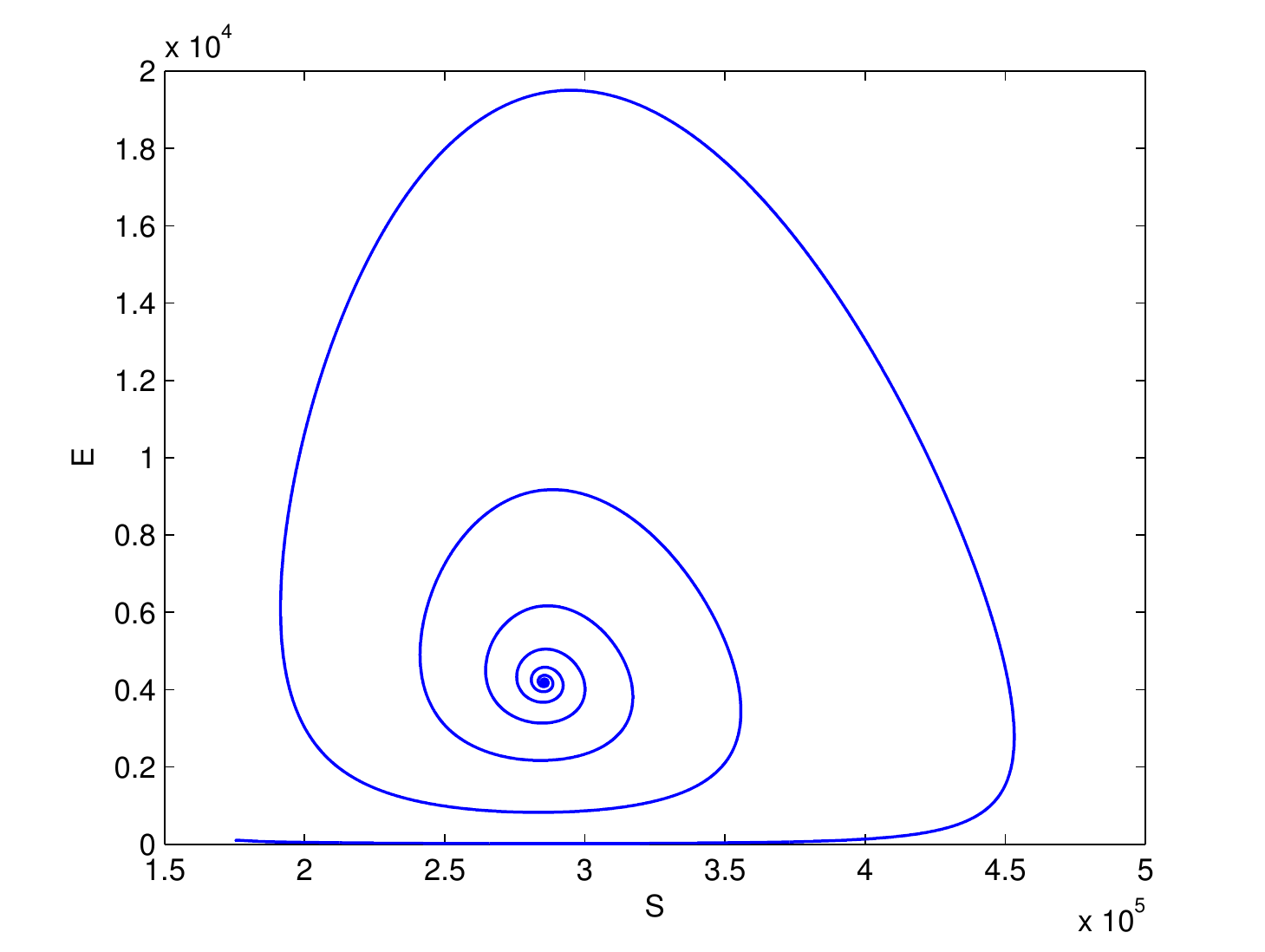}}}
\subfigure[]{\rotatebox{0}{\includegraphics[width=0.45 \textwidth,
height=45mm]{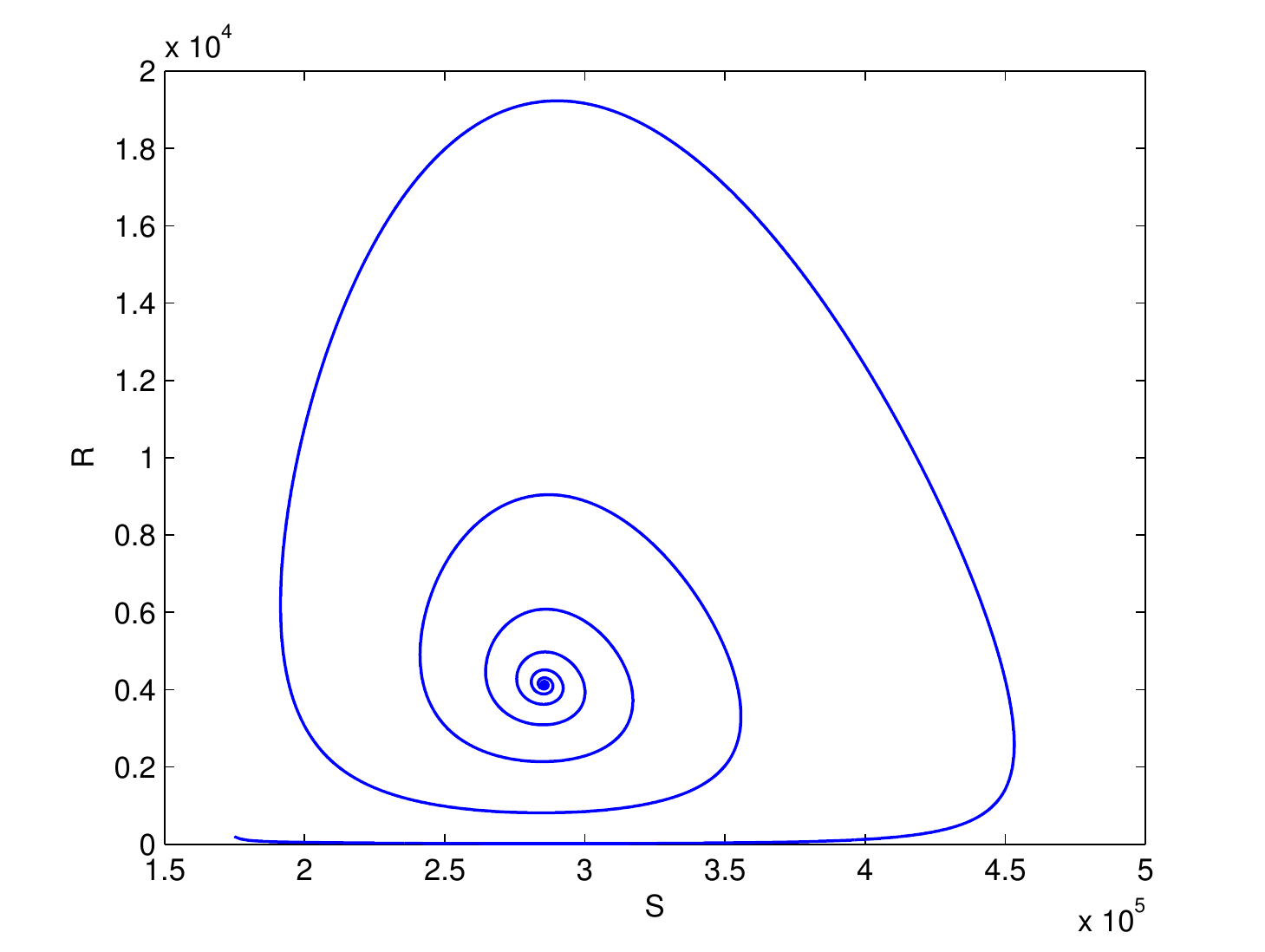}}}

 \vspace{-2mm}
 \caption{
\footnotesize Time history and phase portraits of system (\ref{eqn:modela}) for $\beta=0.0000005$,\,  
$q_E=0.86$,\, 
$d_E=0.0036$,\,
$d_R=0.0036$,\,
$p_E= 0.12$,\,
$p_R=0.12$,\,
$c_E=0.25$,\, 
$c_R=0.15$,\, 
$k=0.56$,\,
$\delta=0.1$,\, 
$\mu =0.000034247$,\,
$\Lambda=  600$ and 
$q_R=0.14$.}\label{fig:E1}
\end{center}
 \end{figure}


\section{Numerical Simulations}\label{sec:6}

In this section, we present some numerical simulations of system \eqref{eqn:modela}  to support our analytical results.

First, we choose $ \beta =0.0000005 $, $q_E  = 0.86 $, $d_E =0.0036 $, $d_R=0.0036 $, $p_E=p_R=0.12$, $c_E=0.25$, $c_R=0.15$, $k=0.56$, $\delta=0.1$, $\mu =0.000034247$, $\Lambda=600$, and $q_R=0.14$.
In this case we find that $ \mathcal{R} _0 =  5.039762256$, and thus, by Theorem \ref{thm:x_ast},  the endemic equilibrium  $ x ^\ast $ is globally asymptotically stable in $ \mathbb{R}  ^4 _{ >0 }$. Figures \ref{fig:E1} (a)-(d) depict $ S, E, R $, and $ T $ as a function of the time $ t $ (days), and show that after a few oscillations these populations approach a constant value. Figures 
\ref{fig:E1} (e) and (f), instead, are phase portraits obtained for different initial conditions. These two figures  confirm that  the solutions approach a globally asymptotically stable equilibrium point. This case illustrates the unwanted scenario where terrorists and recruiters become endemic to the population. 

Second, we increase the rates   $ p _E $ and $ p _R $  at which extremist and recruiters enter the $T$ compartment to   $p_E=p_R=0.92$ and leave the rest of the parameters unchanged. This can be viewed as an improvement of the de-radicalization programs. 
Figures \ref{fig:E0} (a)--(d) show that $ S, E, R, T \to 0  $, as the time  $ t $ grows large, confirming that $ x _0 $ is globally asymptotically stable. This is  the preferred situation, where extremists and recruiters die out in the long run.

 \begin{figure}[!h]
\begin{center}
\subfigure[]{\rotatebox{0}{\includegraphics[width=0.45 \textwidth,
height=45mm]{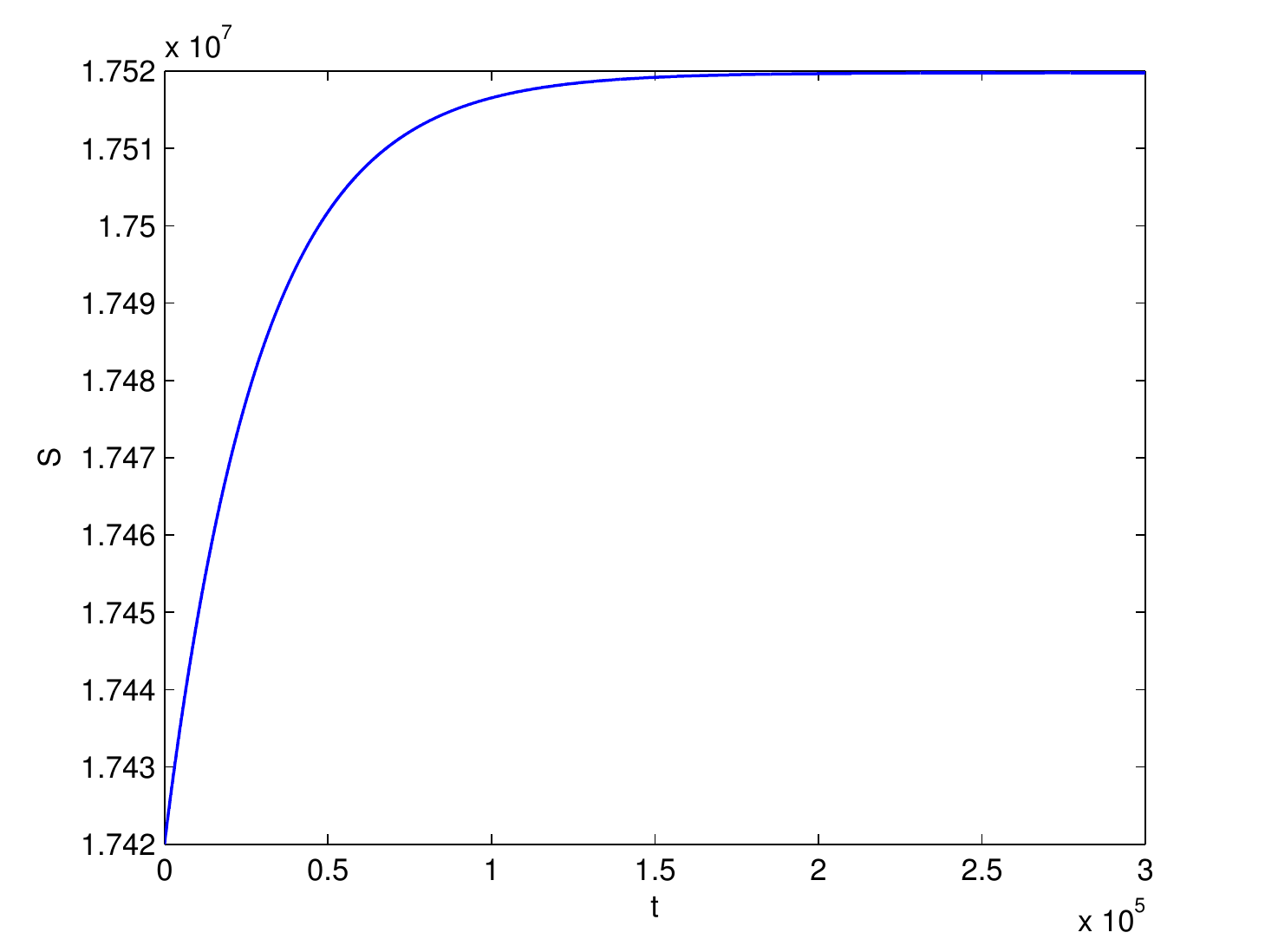}}}
\subfigure[]{\rotatebox{0}{\includegraphics[width=0.45 \textwidth,
height=45mm]{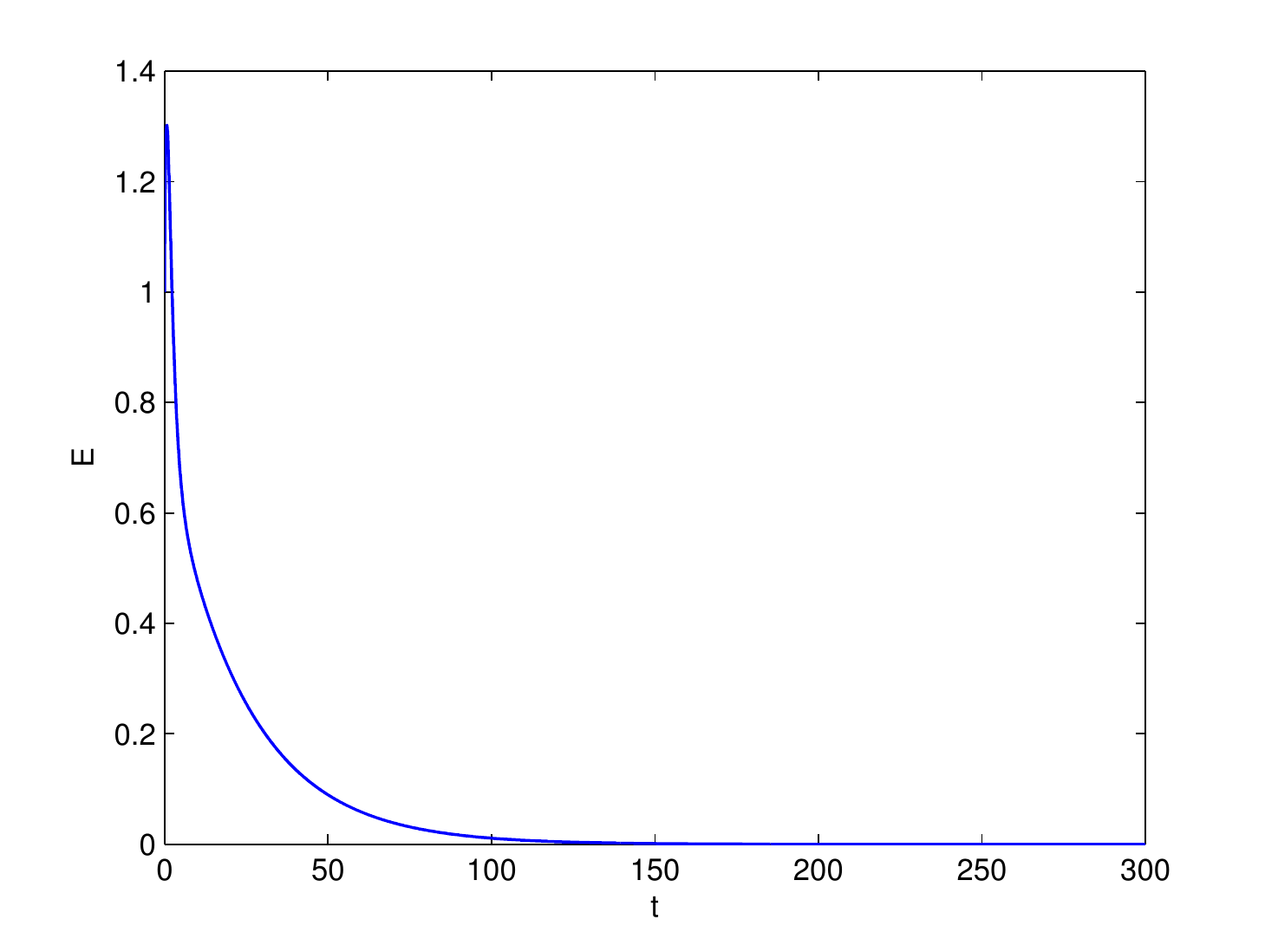}}}
\subfigure[]{\rotatebox{0}{\includegraphics[width=0.45 \textwidth,
height=45mm]{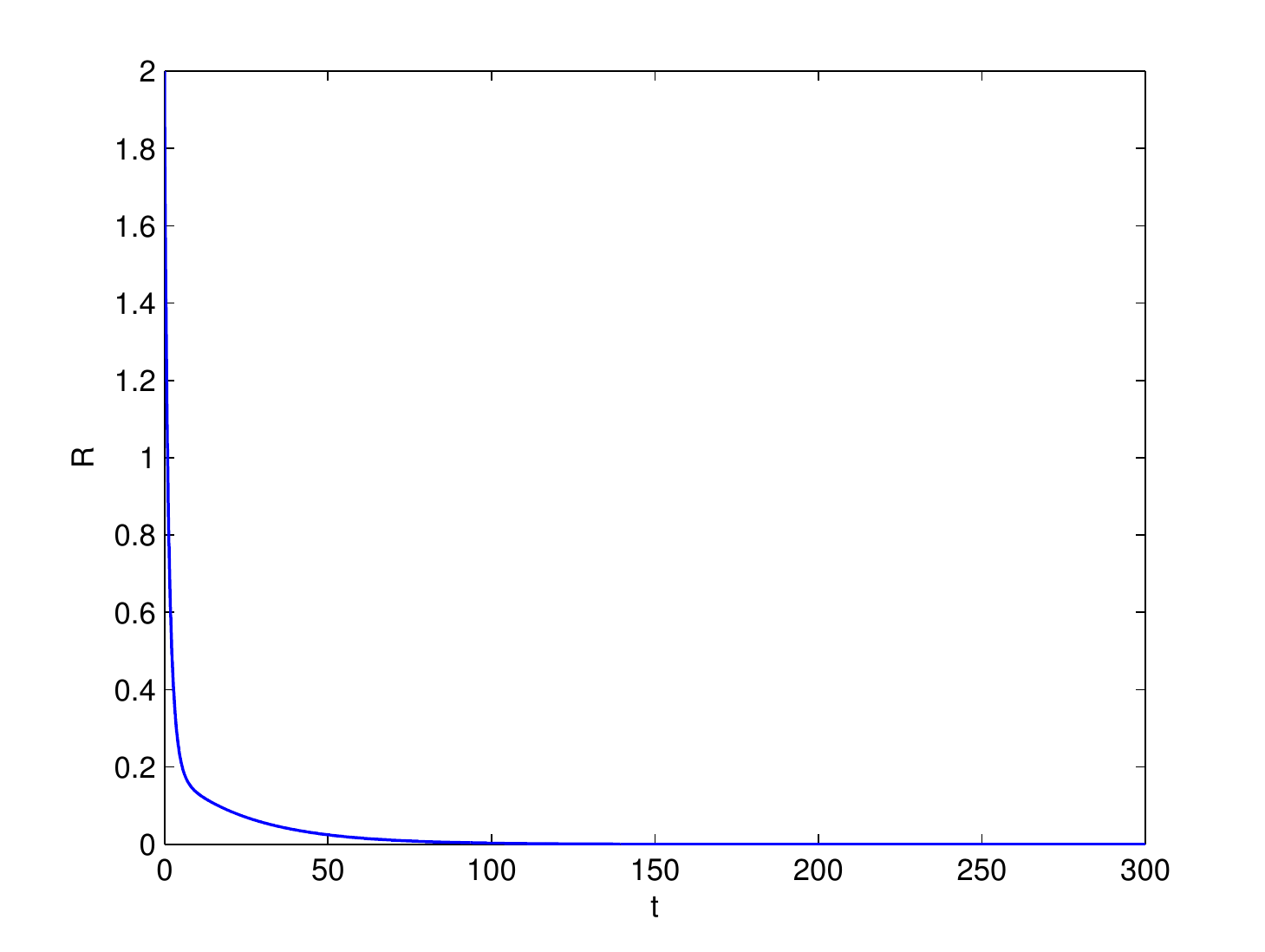}}}
\subfigure[]{\rotatebox{0}{\includegraphics[width=0.45 \textwidth,
height=45mm]{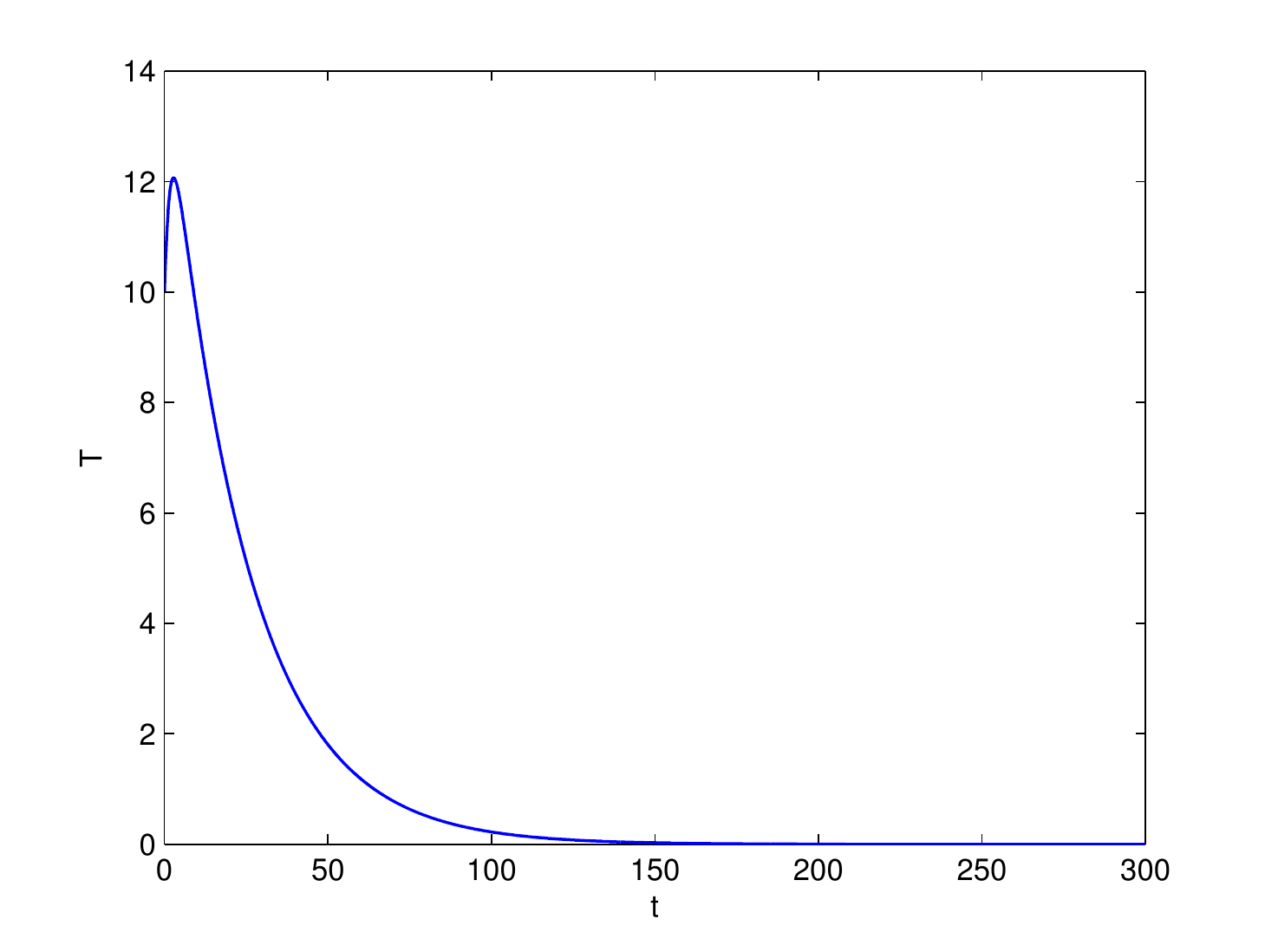}}}

 \vspace{-2mm}
 \caption{
\footnotesize Time history of system (\ref{eqn:modela}) for $\beta=0.00000005$,\,
$q_E=0.86$,\,
$d_E=0.0036$,\,
$d_R=0.0036$,\,
$p_E= 0.92$,\,
$p_R=0.92$,\,
$c_E=0.25$,\,
$c_R=0.15$,\,
$k=0.56$,\,
$\delta=0.1$,\,
$\mu =0.000034247$,\,
$\Lambda=  600$ and 
$q_R=0.14$.}\label{fig:E0}
\end{center}
 \end{figure}
\section{Discussion}\label{sec:7}
In this paper, we modified a  compartmental model of  radicalization proposed by McCluskey and Santoprete \cite{mccluskey2017bare} to include the  deradicalization process. By means of the next generation method we obtained the  basic reproduction number $ \mathcal{R} _0 $, which plays an important role in controlling the spread of the extremist ideology. By constructing two Lyapunov functions we studied  the global stability of the equilibria. We showed that this new model  displays a threshold dynamics. 
When  $ \mathcal{R} _0 \leq 1$ all solutions converge to the radicalization-free equilibrium, and the populations of recruiters and extremists eventually die out. When $ \mathcal{R} _0 > 1 $  the radicalization-free equilibrium is unstable and there is also an additional     endemic equilibrium that is  globally asymptotically stable. In this case extremists and recruiters will persist in the population.  
Since we expressed the basic reproduction number in terms of the parameters of the model we were able to evaluate strategies for countering violent extremism. These strategies were outlined in the introduction.

Of course, when modeling social dynamics one has to make many simplifying assumptions. The model studied in this paper is not completely free from this defect. One issue, for instance, is that extremists and recruiters entering the treatment compartment will stay in the compartment for a period of time, given by the length of the prison sentence or of the de-radicalization treatment. Hence, it seems possible to consider more realistic models by using delay differential equations, and include the time of the de-radicalization treatment as a time delay.  Another issue is that the population in the various compartments may not be homogeneous. For example, the parameter $ \beta $ may depend on the age of the susceptible, suggesting that an age-structured model may be better suited to describe this problem.
We plan to address these and other issues in future studies. 

 

\bibliography{rad_papers}{}
\bibliographystyle{siam}

\end{document}